\documentclass[12pt]{article}

\usepackage{setspace}
\usepackage{color}
\usepackage{amsfonts}
\usepackage{amssymb}
\usepackage{amsmath}
\usepackage[mathscr]{eucal}
\usepackage{amsfonts}
\usepackage{amsmath,amsxtra,latexsym,amsthm, amssymb, amscd}
\usepackage[colorlinks,citecolor=blue,filecolor=black,linkcolor=blue,urlcolor=black]{hyperref}
\usepackage{pgfplots} 
\usepackage{url}
\usepackage[round]{natbib}
\usepackage{fancyhdr}
\usepackage{tikz}

\usepackage{graphicx}
\usepackage{hyperref}
\usepackage{subcaption}
\usepackage{caption}
\usepackage{epstopdf}

\usepackage{soul}

\newtheorem{theorem}{Theorem}

\newtheorem{corollary}{Corollary}

\newtheorem{definition}{Definition}

\newtheorem{lemma}{Lemma}

\newtheorem{proposition}{Proposition}
\newtheorem{remark}{Remark}

\newtheorem{assum}{Assumption}

\newcommand{\rr}{\mathbb{R}}

\newcommand{\ma}{\max\limits}

\newcommand{\limm}{\lim\limits}

\newcommand{\blue}[1]{\textcolor{blue}{#1}}

\begin{document}

\title{Equilibrium with non-convex preferences: some insights\footnote{The authors do not have any conflict of interest to declare.}\footnote{The authors would like to thank two Referees for their helpful comments and suggestions. Their points have helped us to substantially improve our previous version.}
}

\author{Cuong Le Van\footnote{Paris School of Economics, TIMAS, CASED. {Email}: levan@univ-paris1.fr. Address: CES-Centre d'Economie de la Sorbonne, 106-112 boulevard de l'Hôpital, 75647 Paris Cedex 13, France.}\and
Ngoc-Sang Pham\footnote{EM Normandie Business School, M\'etis Lab. {Email}: npham@em-normandie.fr. Address: EM Normandie, 9 Rue Claude Bloch, 14000 Caen, France.}} 
\date{March 07, 2025}


\maketitle
\begin{abstract}
We study the existence of equilibrium when agents' preferences may not be convex. For some specific utility functions, we provide a necessary and sufficient condition under which there exists an equilibrium. The standard approach cannot be directly applied to our examples because the demand correspondence of some agents is neither single-valued nor convex-valued. \\
\noindent {\it JEL Classifications:} D11, D51\\
{\it Keywords:} Non-convex preferences, general equilibrium.

\end{abstract}
\newpage
\section{Introduction}
The existence of equilibrium is one of the fundamental issues in economic theory. 
 When proving the existence of general equilibrium, the standard approach requires the convexity of the preferences (or the {(quasi)}concavity of the utility function).\footnote{\citet{debreu82},  \citet{florenzano03}, \cite{abb89} offer  excellent treatments of the existence of equilibrium.} The existence of equilibrium with non-convex preferences is an important issue not only in  microeconomics but also in finance because  investors' preferences may not be convex. However, this question remains open. Although \cite{Aumann1966} proves the existence of general equilibrium in an exchange economy consisting of a continuum of agents with non-convex preferences,\footnote{A key point in \cite{Aumann1966} is that the aggregate preferred set is convex. He proves this by using a mathematical result which states that the integral of any set-valued function over a non-atomic measure space is convex \citep{Aumann1965, Richter1963}.} he also recognizes that such a result may not hold  when there are finitely many agents.

  Recently, \cite{aa18} study the equilibrium existence in a model with two kinds of agents:  risk averse agents (having strictly concave utility function) and  risk loving agents (having strictly convex utility function).  A key result is  \cite{aa18} is that there exists an equilibrium if the aggregate endowment of risk averse agents  is sufficiently large in some state of nature compared to the aggregate endowment in other states. Moreover, such an equilibrium is a corner equilibrium.

The main aim of our paper is to study the issue of the existence of equilibrium when agents may be neither risk averse nor risk loving (i.e., the agents' utility functions are neither concave nor convex).
 More precisely, we focus on a two-agent two-good exchange economy. The type A agent (risk averse agent) has  a strictly convex utility function while the type B  agent's utility function 
 is neither convex nor concave. We assume that the agent $B$ is risk loving with good 1 and risk averse with good 2.
 

The demand of agent $A$ is single-valued and continuous. However,  the demand of the type B agent may be multiple-valued.  By the way, the standard approach (see, for instance, \cite{mc95} and references therein) cannot be directly applied to our example.  Notice also that the results in \cite{aa18} cannot be applied in our model because  agent B's utility function is neither concave nor convex. 

Under our specifications, we manage to characterize all possible equilibria and find out a necessary and sufficient condition under which equilibrium exists (Proposition \ref{main1}). We can also explicitly compute the equilibrium outcomes. Our necessary and sufficient condition allows us to understand how agents' preferences and endowments affect the equilibrium existence.

First, we look at the role of endowments. If we fix all parameters excepted the endowments ($e^A_1,e^A_2$) of agent A (risk averse agent), then there exists an equilibrium if $e^A_1$ or $e^A_2$ is high enough (see Corollary \ref{ea2-high}). This result is consistent with the main finding in \cite{aa18}. The difference between \cite{aa18} and our paper is that while the equilibrium in \cite{aa18} is a corner equilibrium, the equilibrium in our model may be interior or corner, depending on the distribution of endowments and the preferences of the agents.

More interestingly, we show there does not exist an equilibrium even when the good 2 endowment of the agent $B$ (who is neither risk lover nor risk averter) is very high. This point has not been mentioned by \cite{aa18}.

 
 Second, we explore the role of agents' preferences, i.e., the weight that agent $B$ puts on the risk averse good.  
 Given a distribution of endowments, when the agent $B$ strongly wants to consume the risk averse good, 
   then there exists a unique equilibrium. Moreover, in this equilibrium, agent $B$ does not consume the risk loving good. Although we have the uniqueness of equilibrium, the demand of agent $B$ may be multiple-valued.  By contrast, if  agent $B$ strongly wants the risk loving good, 
then there exists a unique equilibrium, and this equilibrium is interior. Interestingly, there is a room for the non-existence of equilibrium. This happens for some range of agents' preference. 

The paper proceeds as follows. In Section \ref{exchange}, we present the structure of the economy and the concept of general equilibrium.   
In Section \ref{22}, we study the existence and properties of general equilibrium. Section \ref{discuss}  provides discussions and Section \ref{conclusion} concludes. Technical proofs are presented in Appendix \ref{appendix}.

\section{An exchange economy}\label{exchange}
Consider an exchange economy with $L$ goods,  $m$ agents, and the consumption set is $\rr^L_+$ for any agent.

\begin{assum}\label{assumption1}
For $i\in \{1,2,\ldots,m\}$, agent $i$ has utility function $U^i(c_1,\ldots,c_L)$ which is continuous and strictly increasing in each component.  Each agent $i$ has endowments $(e^{i}_1,\ldots,e^i_L)\gg 0$.\footnote{It means that $e^i_l>0$ for any $l\in\{1,\ldots,m\}.$}
\end{assum}

 The aggregate supply of good $l$ is  $e_l\equiv \sum_{i=1}^me^i_l$. 
Given endowments and prices, the income of agent $i$ is $w_i\equiv \sum_{l=1}^Lp_le^i_l$.

We provide a standard definition of general equilibrium. 
\begin{definition} \label{def-equilibrium}
A list $\big(p_1,\ldots, p_L, (c^i_1,\cdots, c^i_L)_{i=1}^m\big)$ of non-negative numbers is a general equilibrium if (i) given $(p_1,\ldots, p_L)$, for each $i \in \{1,\ldots,m\}$, the allocation $(c^i_1,\cdots, c^i_L)$ is a solution to  agent $i'$s maximization problem
\begin{align}\label{ui}
&\max_{(c_1,\ldots, c_L)\geq 0}U^i(c_1,\ldots, c_L)\\
&\text{subject to: }p_1c_1+\cdots+p_Lc_L\leq p_1e^i_1+\cdots+p_Le^i_L
\end{align}
(ii) markets clear: $\sum_{i=1}^mc^i_l=\sum_{i=1}^me^i_l$ for any $l=1,\ldots,L$, and (iii) $p_1>0,\ldots,p_L>0$.
\end{definition}

We state a simple version concerning the equilibrium existence.
\begin{theorem}[Arrow-Debreu]\label{theorem1}
Under Assumption \ref{assumption1} and assume that $U^i$ is strictly concave. Then there exists a general equilibrium.
\end{theorem}
A standard proof of Theorem \ref{theorem1} can be found in Theorem 1.4.9 in \cite{abb89} or  Proposition 17.C.1 in  \citet{mc95}.

In the literature (see \cite{debreu82}, \cite{ florenzano03}, \cite{abb89} among others), many work under more general models (with many commodities, many agents, with preferences instead of utility functions, and with production). Under some assumptions, the literature provides several results proving the existence of a competitive equilibrium.  

In establishing the existence theorems, the convexity of agents' preferences is one of the key assumptions. The main aim of our paper is to understand the role of the convexity of preferences on the equilibrium existence via simple models.

\section{Economy with two agents and two goods}

\label{22}
We now focus on the case where there are two goods,  two agents ($A$ and $B$), and the consumption set is  $\rr^2_+$. We choose this approach to find out not only  conditions ensuring the equilibrium existence but also explicit conditions under which there does not exist any equilibrium.

Given $p_1>0,p_2>0$, the optimization problem of agent $i$ has at least 1 solution. Since the budget constraint is equivalent to $c_1+pc_2\leq e^i_1+pe^i_2 $ where \blue{$p\equiv p_2/p_1$}, the set of solutions only depends on $ e^i_1+pe^i_2$ and $p$. So, we denote the demand of agent $i=A,B$ by  $$z^i(p,e^i_1+pe^i_2)=
\big(z^i_1(p,e^i_1+pe^i_2),z^i_2(p,e^i_1+pe^i_2)\big)\subset \rr_+^2.$$ 

According to Definition \ref{def-equilibrium}, 
$p\equiv p_2/p_1\in (0,\infty)$ is an equilibrium relative price if and only if there exists $c^i(p,e^i_1+pe^i_2)\in z^i(p,e^i_1+pe^i_2)$ for any $i\in \{A,B\}$, so that
\begin{align}
\underbrace{c^A_1(p,e^A_1+pe^A_2)+c^B_1(p,e^B_1+pe^B_2)}_\text{Aggregate demand}=\underbrace{e^A_1+e^B_1.}_\text{Aggregate supply}
\end{align}

We now focus on a case where the demand for good $1$ of agent $A$ is singleton. This holds if the utility function of agent $A$ is strictly concave.

\begin{assum}\label{assumption2}
The demand of agent $A$ is singleton and continuous in the relative price $p$ and endowments $e^A_1, e^A_2$.  In this case, we denote $c^A_i(p,e^A_1+pe^A_2)$ the demand function for good $i=1,2$.
\end{assum}

We state an useful claim.
\begin{lemma}\label{2agents}
Consider a two-agent two-good exchange economy. Let Assumptions \ref{assumption2} and \ref{assumption2} hold.  There exists an equilibrium if and only if there exists $p\in (0,\infty)$ and 
$c^B(p,e^B_1+pe^B_2)\in z^B(p,e^B_1+pe^B_2)$ so that
\begin{align}
c^A_1(p,e^A_1+pe^A_2)+c^B_1(p,e^B_1+pe^B_2)-(e^A_1+e^B_1)=0.
\end{align}
\end{lemma}
\subsection{A tractable case}

For the sake of tractability, we assume that  agent $B$ has utility function $U^B(c_1,c_2)=\frac{c_1^2}{2} + \mathcal{D}\ln(c_2)$, where $\mathcal{D}>0$. This agent is risk loving with good 1 but risk averse with good 2. Note that this function is neither quasiconcave nor quasiconvex on the consumption set $R^2_+$. It means that the preferences are not convex.  A motivating example of this function is that this agent is risk-averse about their health (good 2) but risk-seeking about leisure activities (good 1).\footnote{As \cite{aa18} mentioned, it is well known that an exchange economy with $L$ commodities can be interpreted as a financial economy with $L$ states of nature and with complete financial markets (Arrow-Debreu securites).}

To present the demand of agent $B$, let us denote $x^*$ the unique solution {to} the equation $g(x)=0$ where\footnote{We can compute that $
 g^\prime (x)= \frac{x+\sqrt {x^2 -4 } }{\sqrt{x^2 -4 }} \Big(\frac{x+\sqrt {x^2 -4 } }{4}-\frac{1}{x}\Big)>0$ 
when  $x>2$. So, $g$ is strictly increasing on $[2,\infty)$. Moreover, $g(2)<0<g(\infty)$. So, the function $g$ has a unique solution.}
\begin{align}
g(x)\equiv \frac{1}{8} (x+\sqrt {x^2 -4})^2 + \ln \Big(1-\sqrt{1-\frac{4}{x^{2}}} \Big)-\ln(2).
\end{align}
Note that \blue{$x^*\approx 2.2$}.

We have the following result showing the demand of agent $B$.

\begin{lemma}\label{agentB}
Let $p_1,p_2>0$. Denote $p\equiv p_2/p_1$. 
The demand correspondence for good 1 of agent $B$ is given by
\begin{align}
\label{cb0}c^B_1=\begin{cases}\{0\} &\text{ if } e^B_1+pe^B
_2<x^*\sqrt{\mathcal{D}} \\
 \left \{0, e^B_1+pe^B_2+\sqrt{(e^B_1+pe^B_2)^2-4\mathcal{D}}\right \}& \text{ if }e^B_1+pe^B_2=x^*\sqrt{\mathcal{D}} \\
 \Big\{e^B_1+pe^B_2+\sqrt{(e^B_1+pe^B_2)^2-4\mathcal{D}}\Big\} &\text{ if } e^B_1+pe^B_2>x^*\sqrt{\mathcal{D}}.
\end{cases}
\end{align}
\end{lemma}
\begin{proof}
See Appendix \ref{appendix}.
\end{proof}
In the first case in (\ref{cb0}), the real revenue (in term of good 1) $e^B_1+pe^b_2$ is low which implies that   agent $B$ does not consume good $1$. When the revenue is high enough (case 3), he(she) consumes  both goods.

It should be noticed that the demand correspondence of agent $B$ is single-valued everywhere excepted when $e^B_1+pe^B_2=x^*\sqrt{\mathcal{D}}$. It is upper semi-continuous, compact-valued but not convex-valued.\footnote{The demand correspondence (\ref{cb0}) is one of the simplest correspondences which are multiple-valued but not convex-valued.} By consequence, conditions needed in the standard proof of the equilibrium existence are violated. Indeed, firstly, we cannot directly apply Theorem 1.4.9 in \cite{abb89} or  Proposition 17.C.1 in \citet{mc95}  to our example because the demand correspondence of agent $B$ is double-valued when $e^B_1+pe^B_2=x^*\sqrt{\mathcal{D}}$. Secondly, since the aggregate demand correspondence is not convex-valued, the method in \cite{McKenzie59} cannot be applied as well.

Our goal is to investigate conditions under which there exists a general equilibrium or not. We have the following result which is a direct consequence of Lemmas  \ref{2agents} and \ref{agentB}.

\begin{proposition}[(non)existence of equilibrium]\label{mainlemma} 
Let us consider a two-good two-agent exchange economy. Let Assumption \ref{assumption2} be satisfied.

 Assume that  $U^B(c_1,c_2)=\frac{c_1^2}{2} + \mathcal{D}\ln(c_2)$, where $\mathcal{D}>0$. 
\begin{enumerate}
\item There exists an equilibrium $(p_1,p_2, c^A _1, c^A_2, c^B _1, c^B _2)$  with $c^B _1=0$ and the relative price $p=p_2/p_1$ if there exists $p$ satisfies
\begin{subequations}
 \begin{align}\label{pcor}
&c^A_i(p,e^A_1+pe^A_2)-(e^A_1+e^B_1)=0\\
&e^B_1+pe^B_2< x^*\sqrt{\mathcal{D}}.
\end{align} 
\end{subequations}
\item There exists an equilibrium $(p_1,p_2, c^A _1, c^A_2, c^B _1, c^B _2)$  with $c^B _1>0$ and the relative price $p=p_2/p_1$ if there exists $p$ satisfies
\begin{subequations}
 \begin{align}
&c^A_i(p,e^A_1+pe^A_2)+\Big(e^B_1+pe^B_2+\sqrt{(e^B_1+pe^B_2)^2-4\mathcal{D}}\Big)-(e^A_1+e^B_1)=0\\
 &\label{cond-interior1}e^B_1+pe^B_2> x^*\sqrt{\mathcal{D}}.
\end{align}
\end{subequations}
\item There is no equilibrium if the two following conditions are satisfied.
\begin{enumerate}
\item The equation $Z_{cor}(p)\equiv c^A_i(p,e^A_1+pe^A_2)-(e^A_1+e^B_1)=0$ does not have a solution satisfying $e^B_1+pe^B_2\leq x^*\sqrt{\mathcal{D}}$.
\item The equation $Z_{int}(p)\equiv c^A_i(p,e^A_1+pe^B_2)+\Big(e^B_1+pe^B_2+\sqrt{(e^B_1+pe^B_2)^2-4\mathcal{D}}\Big)-(e^A_1+e^B_1)=0$ does not have a solution satisfying $e^B_1+pe^B_2\geq x^*\sqrt{\mathcal{D}}$.
\end{enumerate}
\end{enumerate}
\end{proposition}

Although Proposition \ref{mainlemma} shows important insights, its conclusion depends on the form of the demand functions of agent $A$. The last point in Proposition \ref{mainlemma} suggests that there may not  exist any equilibrium if $\mathcal{D}$ is in some interval which depends on agents' endowments.

To have a tractable model where we can easily compute the equilibrium price, we assume that $U^A(c_1,c_2)=\ln(c_1)+\ln(c_2)$. We choose this specification because it satisfies Inada's condition which helps us to avoid considering several cases where $c^A_1$ can be zero.

Let us denote \blue{$\pi^{cor}\equiv \frac{2e^B_1+e^A_1}{e^A_2}$} and \blue{$\pi^{int}$ the smallest root} (if there exists a root) of the function
\begin{align*}
F(X)\equiv& \big[(e^B_2+e^A_2)^2-(e^B_2)^2\big]X^2-2\big[(e^B_1+e^A_1)(e^B_2+e^A_2)+e^B_1e^B_2\big]X\\
&+(e^B_1+e^A_1)^2+4\mathcal{D}-(e_1^B)^2.
\end{align*}
Precisely, if $\Delta\equiv \big[(e^B_1+e^A_1)(e^B_2+e^A_2)+e^B_1e^B_2\big]^2-\big[(e^B_2+e^A_2)^2-(e^B_2)^2\big]\big[(e^B_1+e^A_1)^2+4\mathcal{D}-(e_1^B)^2\big]\geq 0$, we define 
\begin{align}\label{pint}
\pi^{int}\equiv \frac{(e^B_1+e^A_1)(e^B_2+e^A_2)+e^B_1e^B_2-\sqrt{\Delta}}{(e^B_2+e^A_2)^2-(e^B_2)^2}.
\end{align} 
In this case, observe that 
\begin{align}\label{rank}
 0<\pi^{int} <\frac{e^B_1+e^A_1}{e^B_2+e^A_2}< X^*<\frac{2e^B_1+e^A_1}{e^A_2}=\pi^{cor}
\end{align}
where $X^*\equiv \frac{(e^B_1+e^A_1)(e^B_2+e^A_2)+e^B_1e^B_2}{(e^B_2+e^A_2)^2-(e^B_2)^2}$ which satisfies $F^\prime(X^*)=0$.

By Inada condition, $c^A_1, c^A_2$ and $c^B_2$ are  strictly positive {while} $c^B_1$ may be zero or strictly positive  {at} equilibrium. Our main result is to proved a full characterization of general equilibrium.

\begin{proposition}[existence and computation of equilibrium]\label{main1} 
Let us consider a two-good two-agent exchange economy with $U^A(c_1,c_2)=\ln(c_1)+\ln(c_2)$  and $U^B(c_1,c_2)=\frac{c_1^2}{2} + \mathcal{D}\ln(c_2)$, where $\mathcal{D}>0$.
\begin{enumerate}
\item 

There exists an equilibrium $(p_1,p_2, c^A _1, c^A_2, c^B _1, c^B _2)$  with $c^B _1=0$ if and only if  \begin{align}\label{pcor}
e^B_1+e^B_2\frac{2e^B_1+e^A_1}{e^A_2}\leq x^*\sqrt{\mathcal{D}}.
\end{align} Such an equilibrium is unique, up to a positive scalar for the prices. The equilibrium relative price $p_2/p_1= \pi^{cor}\equiv \frac{2e^B_1+e^A_1}{e^A_2}$.

\item There exists an equilibrium $(p_1,p_2, c^A _1, c^A_2, c^B _1, c^B _2)$  with $c^B _1>0$ if and only if\footnote{We can prove that $\Delta\geq 0$ is equivalent to $\frac{\big[(e^A_1+e^B_1)e^B_2+(e^A_2+e^B_2)e^B_1\big]^2}{e^A_2(e^A_2+2e^B_2)}\geq 4\mathcal{D}$ meaning that the weigh $\mathcal{D}$ should not be too high.}
\begin{align}
 \label{cond-interior1} \frac{\big[(e^A_1+e^B_1)e^B_2+(e^A_2+e^B_2)e^B_1\big]^2}{e^A_2(e^A_2+2e^B_2)}\geq 4\mathcal{D} \text{ and }e^B_1+e^B_2\pi^{int}\geq x^*\sqrt{\mathcal{D}} 
\end{align}
Such an equilibrium is unique, up to a positive scalar for the prices. The equilibrium relative price $p_2/p_1=\pi^{int}$ and $c^B_1=\frac{1}{2}\Big(e^B_1+\pi^{int}e^B_2+\sqrt{(e^B_1+\pi^{int}e^B_2)^2-4\mathcal{D}}\Big)>0$.

\item There is no equilibrium if and only if one of the two following conditions hold: 
\begin{enumerate}
\item \label{3a} $e^B_1+\pi^{cor}e^B_2> x^*\sqrt{\mathcal{D}}$ and $4\mathcal{D}>\frac{\big[(e^A_1+e^B_1)e^B_2+(e^A_2+e^B_2)e^B_1\big]^2}{e^A_2(e^A_2+2e^B_2)}.$

\item \label{3b}$e^B_1+\pi^{cor}e^B_2> x^*\sqrt{\mathcal{D}}> e^B_1+e^B_2\pi^{int}$ and $\frac{\big[(e^A_1+e^B_1)e^B_2+(e^A_2+e^B_2)e^B_1\big]^2}{e^A_2(e^A_2+2e^B_2)}\geq 4\mathcal{D}$
\end{enumerate}

\end{enumerate}
\end{proposition}
\begin{proof}See Appendix \ref{appendix}.\end{proof}

\begin{figure}[ht!]
\centering
\includegraphics[width=7cm,height=6cm]{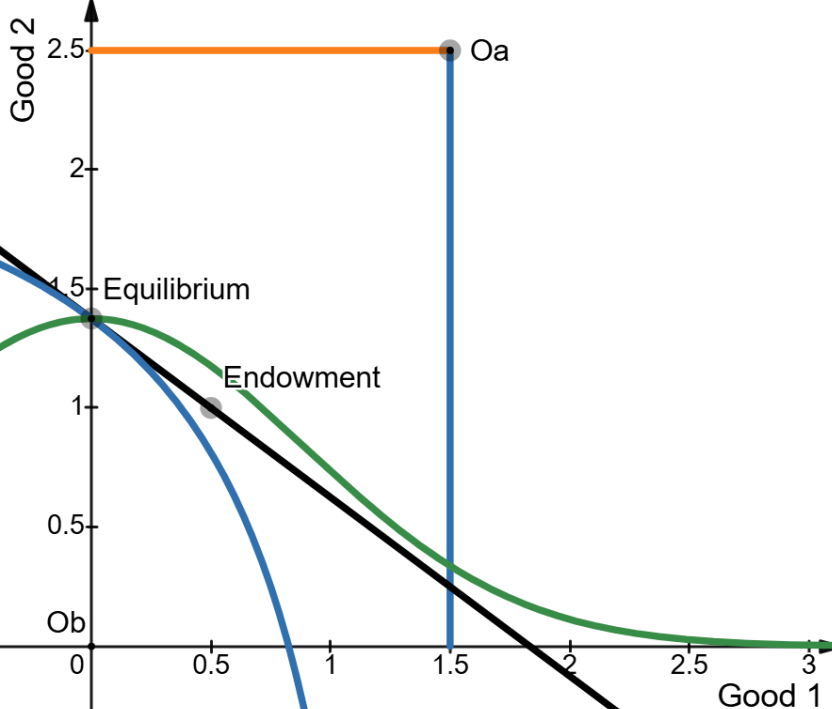}\quad\includegraphics[width=7cm,height=6cm]{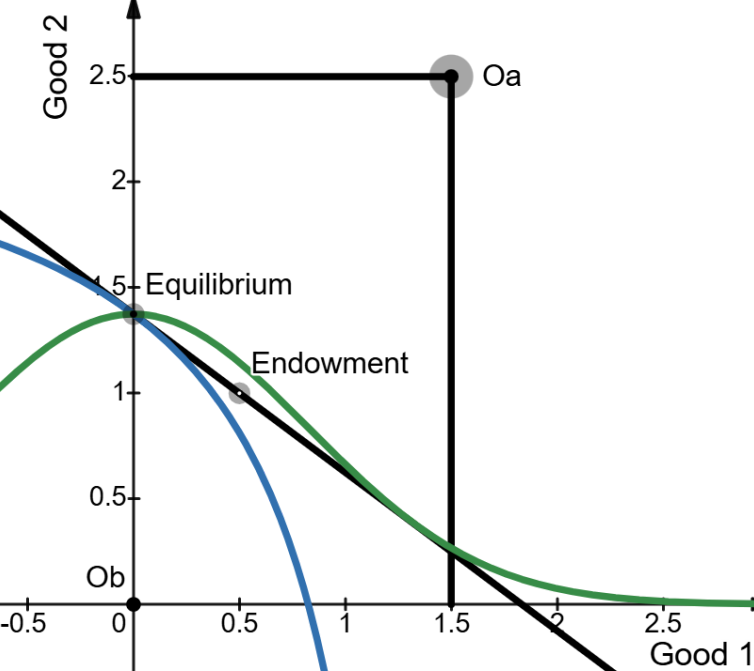}
\caption{\small  Edgeworth box with the {\bf unique corner equilibrium}. LHS: the unique equilibrium. RHS: the unique equilibrium but the  demand for good 1 of agent $B$ is double-valued. The blue (resp., green) curve is the indifferent curve of agent  $A$ (resp., agent $B$) while the black line represents the equilibrium price.}
\label{Fig1}
\end{figure}

\begin{figure}[ht!]
\centering
\includegraphics[width=7cm,height=6cm]{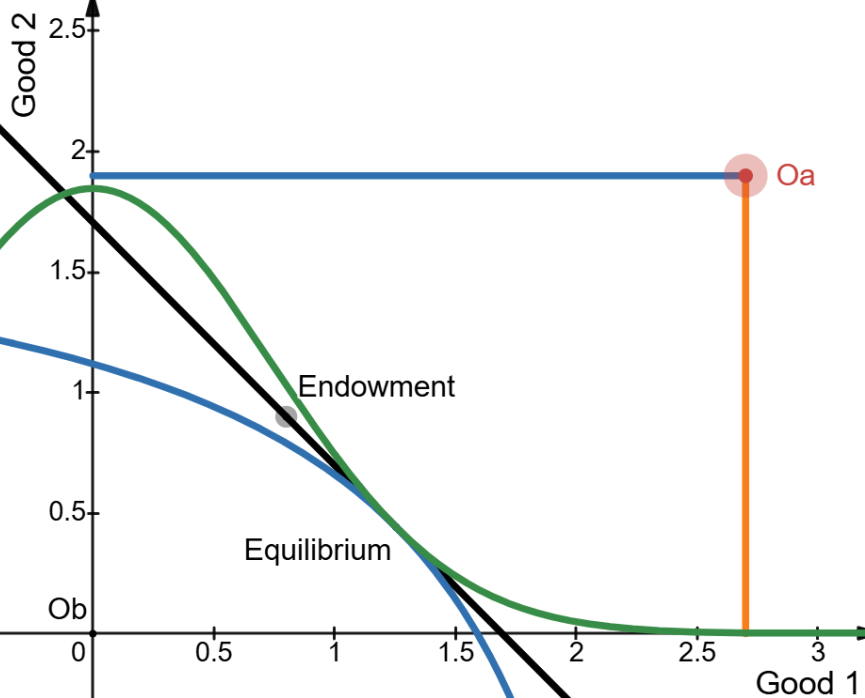}\quad\includegraphics[width=7cm,height=6cm]{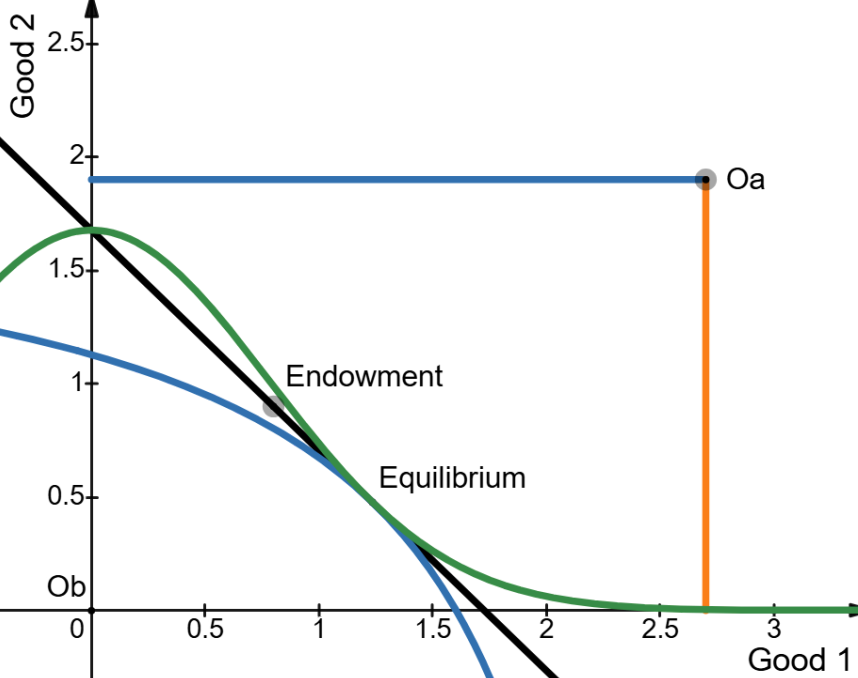}
\caption{\small Edgeworth box with the {\bf unique interior equilibrium}. LHS: the unique equilibrium, RHS: the unique equilibrium but the demand for good 1 of agent $B$ is double-valued.}
\label{Fig2}
\end{figure}


Let us explain the intuition of Proposition \ref{main1}.  Condition  (\ref{pcor}), i.e., $e^B_1+e^B_2\pi^{cor}\leq x^*\sqrt{\mathcal{D}}$  means that the income in terms of good $1$, {when } the relative price $p_2/p_1$ is $\pi^{cor}$, is low with respect to $x^*\sqrt{\mathcal{D}}$. Condition (\ref{cond-interior1}) means that the weight $\mathcal{D}$ should not be high and  that the income $e^B_1+e^B_2\pi^{int}\geq x^*\sqrt{\mathcal{D}}$ in terms of good $1$ of agent $B$, {when} the relative price $p_2/p_1$ is $\pi^{int}$, is high with respect to $x^*\sqrt{\mathcal{D}}$.

Since $x^*\sqrt{\mathcal{D}}$ is  an increasing function of $\mathcal{D}$, condition  (\ref{pcor}) (resp., (\ref{cond-interior1})) {is }satisfied if, and only if,  $\mathcal{D}$ is high enough (resp., low enough). The intuition is that when $\mathcal{D}$ is high enough, agent B strongly wants to consume good $2$. In this case, there exists an equilibrium in which she does not buy good $1$, i.e., $c^B_1=0$. By contrast, when  $\mathcal{D}$ is low enough but still strictly positive, agent $B$ consumes good $1$ at equilibrium.\footnote{Notice that when $\mathcal{D}=0$ (agent $B$ only wants to consume good $1$), there is a unique equilibrium. At equilibrium, the relative price $p_2/p_1=\pi^{int}=\frac{e^A_1}{2e^B_2+e^A_2}$.  
}

Figures \ref{Fig1} and \ref{Fig2} illustrate the unique equilibrium by using the Edgeworth box.\footnote{We numerically draw our figures by using the website https://www.desmos.com/calculator.} Although there is a unique equilibrium, the demand of agent $B$ may be singleton or double-valued (see the right hand side of both Figures \ref{Fig1} and \ref{Fig2}).
 


\subsubsection{Non-existence of equilibrium: role of preferences}\label{nonexistence}
Proposition \ref{main1}'s last point  provides a necessary and sufficient condition for the non-existence of equilibrium.

The intuition behind the non-existence of equilibrium is the following. If there exists an equilibrium, the equilibrium relative price must be either $\pi^{cor}$ or $\pi^{int}$.  However, $\pi^{cor}$ cannot be an equilibrium price because $e^B_1+\pi^{cor}e^B_2> x^*\sqrt{\mathcal{D}}$ meaning that and the relative income (under  the price $\pi^{cor}$) is quite high with respect to the weight $\mathcal{D}$. So, it is not optimal for agent $B$ if he does not consume good 1. In Figure \ref{Fig3} (dashed curves), we observe that point C is not an equilibrium.\footnote{We draw Figure \ref{Fig3} with the following parameters: $e^A_1=e^A_2=1$, $e^B_1=0.8, e^B_2=1$, $\mathcal{D}=0.9$.}

We also see that $\pi^{int}$ cannot be an equilibrium price. Indeed, continuous curves in Figure \ref{Fig3} indicates that point $N$ is not an equilibrium because no agent can find their optimal allocation.
 
 \begin{figure}[ht!]
\centering
\includegraphics[width=11cm,height=8cm]{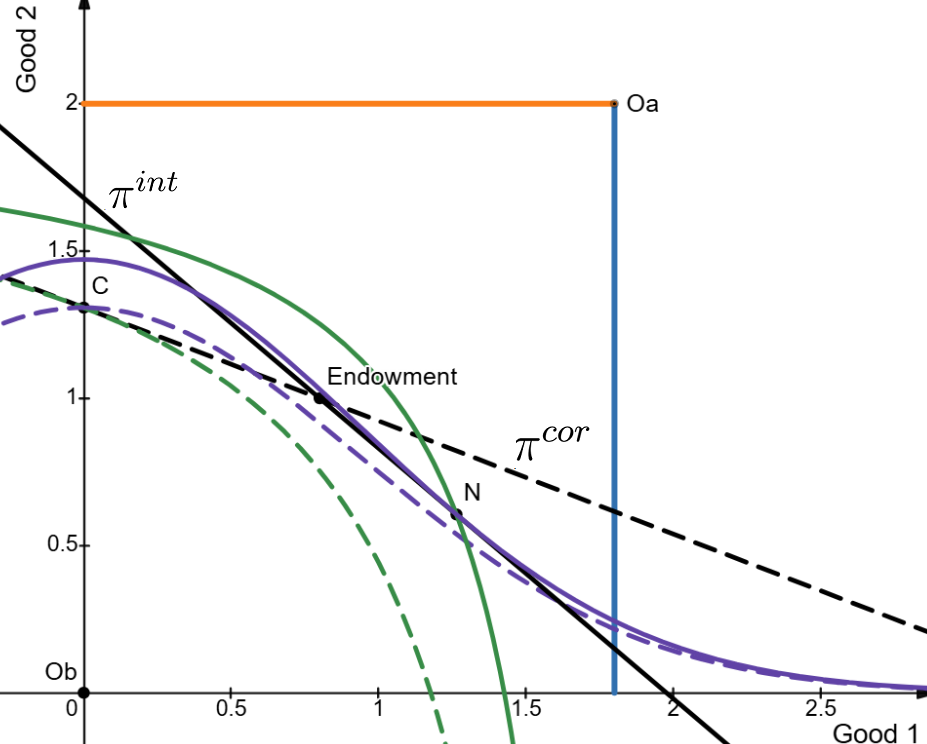}
\caption{\small Edgeworth box. There is {\bf no equilibrium}}
\label{Fig3}
\end{figure}

We can rewrite, for example, condition in the case \ref{3a} of Proposition  \ref{main1} as follows
\begin{align}\label{nonexistence2}
\frac{(e^A_1+e^B_1)e^B_2+(e^A_2+e^B_2)e^B_1}{e^A_2}\frac{1}{x^*}>\sqrt{\mathcal{D}}>\frac{1}{2}\frac{(e^A_1+e^B_1)e^B_2+(e^A_2+e^B_2)e^B_1}{\sqrt{e^A_2(e^A_2+2e^B_2)}}
\end{align}
where recall that $x^*=2.2$. It means that the parameter $\mathcal{D}$ has a middle level. This result is more explicit than point 3 of Proposition \ref{mainlemma}.

\subsubsection{Role of agents' endowments}
Proposition \ref{main1} allows us to understand the roles of agents' endowments on the existence of equilibrium.  We start with the following result.

\begin{corollary}[role of the risk averse agent's endowments] \label{ea2-high} Let Assumptions in Proposition \ref{main1} be satisfied. 
We observe that 
\begin{align}\limm_{e^A_2\rightarrow \infty}\pi^{int}&=\limm_{e^A_2\rightarrow \infty}\pi^{cor}=0\\
\limm_{e^A_1\rightarrow \infty}\pi^{int}&=\infty \text{ and }\limm_{e^A_1\rightarrow \infty}\pi^{cor}=\infty.
\end{align}
\begin{enumerate}\item There exists an equilibrium for any $e^A_2$ high enough (because point 3 in Proposition \ref{main1} cannot happen when $e^A_2$ is high enough).  More precisely, we have that:
\begin{enumerate}
\item If $x^*\sqrt{\mathcal{D}}> e ^B_1$, then when $e^A _2$ is very large,  there exist a unique and $c^B_1=0$  at  equilibrium.
\item If $x^*\sqrt{\mathcal{D}} \leq e^B_1$, then when $e^A _2$ is very large, there exist a unique  and $c^B_1>0$ at equilibrium.
\end{enumerate}
\item When $e^A_1$ is high enough,  there exist a unique equilibrium  and $c^B_1>0$  at  equilibrium.
\end{enumerate}
\end{corollary}

According to Corollary \ref{ea2-high}, there exists an equilibrium if the endowment of good $1$ or good $2$ of the risk-averse agent is high enough. This point is consistent with the main finding in \cite{aa18}. Notice that \cite{aa18} consider general utility functions but that of type $A$ agent is concave and that of type $B$ agent is convex. Their main result is to prove that there exists an equilibrium when the endowment of risk averse agent $e^A_2$ or $e^A_1$ is high enough. 

Corollary \ref{ea2-high} is distinct from \cite{aa18} in two ways.  First, although we work with specific preferences, the utility function of agent $B$ is neither concave nor convex.  By consequence, the method of \cite{aa18} cannot be applied to our model.   Second, in  \cite{aa18}, the optimal allocation of  the type $B$ agent in equilibrium is always in the corner (which corresponds to the case $c^{B}_1=0$ in our model) because this agent's utility is strictly convex. By contrast, in our model, when the risk-aversion agent's endowment is high enough, the equilibrium may be interior. Indeed, this happens if (i) $e^A_2$ is high enough and $x ^* (\mathcal D) \leq e^B_1$ (see point 1.b of Corollary \ref{ea2-high}) or (ii) $e^A_1$ is high enough (see point 2 of Corollary \ref{ea2-high}).

We now show the role of endowments of agent $B$ who is neither risk lover nor risk averter.
\begin{corollary}[role of agent B's endowments]\label{endowments-B}Let Assumptions in Proposition \ref{main1} be satisfied. 
\begin{enumerate}

\item When $e^B_1$ is high enough, condition (\ref{cond-interior1}) holds. Thus, there exists a unique equilibrium and $c^B_1>0$ at equilibrium.
\item We have $\lim_{e^B_2 \to \infty}\big(e^B_1+e^B_2\pi^{int}\big)=\frac{e^A_1+2e^B_1}{2}+\frac{2\mathcal{D}}{e^A_1+2e^B_1}$. So, when $e^B_2$ is high enough, we have that:
\begin{enumerate}
\item If 
$\frac{e^A_1+2e^B_1}{2}+\frac{2\mathcal{D}}{e^A_1+2e^B_1}<x^*\sqrt{\mathcal{D}}$,\footnote{This is satisfied, for example, if $e^A_1+2e^B_1=2\sqrt{\mathcal{D}}$ because $2\sqrt{\mathcal{D}}<2.2\sqrt{\mathcal{D}}=x^*\sqrt{\mathcal{D}}$.} then there is no equilibrium.

\item If $\frac{e^A_1+2e^B_1}{2}+\frac{2\mathcal{D}}{e^A_1+2e^B_1}>x^*\sqrt{\mathcal{D}}$, then there exists a unique equilibrium and $c^B_1>0$ at equilibrium.
\end{enumerate}
\end{enumerate}

\end{corollary}


Interestingly, point 2.a of Corollary \ref{endowments-B} indicates that there may not exist an equilibrium when the good $2$ endowment of agent B (whose utility function is neither concave nor convex) is very high, given that the remaining agent has  concave utility function.  This point complements the main finding of \cite{aa18}.

\begin{remark}\label{remark}We can prove that the equilibrium existence may fail when (1) the good $2$ endowment of the agent $B$ is very high and (2) the utility function of the remaining agent is convex. Indeed, consider an economy consisting of agent $B$ and agent $D$ (who utility function is $U^D(c_1,c_2)=c_2$ which is convex). As we prove in Appendix \ref{appendix} that, there exists an equilibrium if and only if $e^B_1+e^D_1>\mathcal D$ and $
e^B_1+e^D_1+\frac{\mathcal D}{e^B_1+e^D_1}\geq x^*(\mathcal D)$. So, there exists an equilibrium when the good $1$ endowment $e^B_1$ of agent $B$ is high enough but the equilibrium may  fail when the good $2$ endowment $e^B_2$ of agent $B$ is high enough. This is consistent with point 2.a of Corollary \ref{endowments-B}. By consequence, the main result in \cite{aa18} (i.e., if the aggregate endowment of some good of risk averter  is sufficiently large
compared to the aggregate endowment of other good, there is an equilibrium for the economy) may not be true when agents are no longer averter. 
\end{remark}

\section{Discussion and conclusion}
\label{discuss}
In this section, we posit our paper in the existing literature and provide some potential extensions.
\subsection{Other related literatures}
A substantial body of literature examines the existence of competitive equilibrium in the absence of convexity. Since our focus is on an exchange economy, we do not provide here a detailed comparison with the literature on general equilibrium with non-convex technologies.\footnote{There is an extensive literature addressing equilibrium existence in the presence of non-convex technologies and/or externalities. In his seminal work, \cite{guesnerie75} makes use of the notion of the 'cone of interior displacement' and employs the marginal pricing rule in a finite-dimensional setting to characterize Pareto-optimal allocations (Theorem 1), establish conditions for the existence of a PA equilibrium (Theorem 2), and analyze the existence and optimality of QA equilibria. This line of research has been further developed by \cite{cornet90} and others (see \cite{cornet88}, \cite{brown91} for early surveys). More recently, \cite{bf24} prove the existence of a marginal pricing economic equilibrium in a model incorporating increasing returns and externalities, where the commodity space is a Riesz space. For a comprehensive survey on marginal pricing equilibrium with externalities, see the introduction in \cite{bf24}.} Instead, our analysis concentrates on the convexity of consumers' preferences.


Looking back to history, as mentioned by \cite{Aumann1966}, in a setting with a continuum of agents, an equilibrium may exist even preferences are not convex. However, the question remains open in the case of finitely many agents. In an exchange economy with a finite number of agent,  \cite{starr69} proves that, for $\epsilon>0$ given, there exists an $\epsilon$-equilibrium when the number of agents is high enough. 


While a vast body of literature establishes the existence of equilibrium in various settings, significantly fewer studies explore the conditions under which equilibrium fails to exist—an issue that we consider to be of importance. Due to the simplicity of our model, it remains analytically tractable, allowing for a complete characterization of the equilibrium set, including cases of non-existence. Our results (Proposition \ref{main1} and Section \ref{nonexistence}) indicate that the conditions leading to equilibrium failure are not merely exceptional. Indeed, if we take $e^i_l=e>0$ for any $i=A,B$, for any $l=1,2$, then, according to (\ref{nonexistence2}), there exists no equilibrium if $4>\frac{x^*\sqrt{\mathcal{D}}}{e}>2$, i.e., $1.8>\frac{\sqrt{\mathcal{D}}}{e}>0.9$ (because $x^*\approx 2.2)$. This suggests that, in the absence of convexity, the existence of equilibrium hinges on a delicate interplay between the distribution of endowments and agents' preferences. We see this as a promising avenue for future research.

A substantial body of literature also examines equilibrium existence under non-transitive and non-complete preferences. The seminal work of \cite{mascolell74} establishes the existence of a Walrasian equilibrium in a finite-dimensional setting without assuming completeness or transitivity of preferences (a concept also referred to as "non-transitive equilibrium" in Chapter 3 of \cite{florenzano03}). This result has been further extended by several scholars, including \cite{ss75}, \cite{gmc75,gmc79}.

Recently, there is a significant progress establishing the existence of equilibrium in (in)finite dimensional settings which require very minimal assumptions - even in the absence of continuity of preferences (see \cite{hy16}, \cite{ku21}, \cite{py22, py24}, \cite{cornet20}, \cite{adku21, adku22} and references therein). However, these papers still require some versions of convexity of correspondences or preferences, which are not satisfied in our simple model with non-concave utility. For instance,  the demand correspondence (\ref{cb0}) is non-convex-valued and does not have a convex-valued selection (\cite{mascolell74}'s Appendix presents an example of a preference relation on $\rr_+^2$  which satisfies this property).  Furthermore, it does not admit a continuous selection (see Definition 17.62 in \cite{ab06}). By the way, our modest paper contributes to understand not only the existence but also the non-existence of equilibrium.

\subsection{Other forms of utility function}
{\bf Utility function of risk-averse agents}. In our two-good two-agent model, the technique used in Propositions \ref{mainlemma} and \ref{main1} can be applied to other forms of utility of agent $A$. An obvious extension is for the case $U^A(c_1,c_2)=\lambda ln(c_1)+(1-\lambda)ln(c_2)$ where $\lambda\in (0,1)$. In this case, most of results remain the same excepted the value $x^*$ which now depends on $\lambda$.

We now consider a CRRA function $U^A: \mathbb{R}^2_{+} \to \mathbb{R}_{+}$ defined by 
 $u^A(c_1,c_2)=a_1\frac{c_1^{\alpha}}{\alpha}+a_2\frac{c_2^{\alpha}}{\alpha}$ where $a_1>0,a_2>0,\alpha<1, \alpha\not=0$.  In this case, we can find the demand
 \begin{align*}
c_1^A=&(p_1e^A_1+p_2e^A_2)\frac{(\frac{a_1}{p_1})^{\frac{1}{1-\alpha}}}{\frac{a_1^{\frac{1}{1-\alpha}}}{p_1^{\frac{\alpha}{1-\alpha}}}+\frac{a_2^{\frac{1}{1-\alpha}}}{p_2^{\frac{\alpha}{1-\alpha}}}}=\Big(e^A_1+\frac{p_2}{p_1}e^A_2\Big)\frac{a_1^{\frac{1}{1-\alpha}}(\frac{p_2}{p_1})^{\frac{\alpha}{1-\alpha}}}{a_1^{\frac{1}{1-\alpha}}(\frac{p_2}{p_1})^{\frac{\alpha}{1-\alpha}}+a_2^{\frac{1}{1-\alpha}}}.
 \end{align*}
Then, we can apply Proposition \ref{mainlemma} to get a similar result but less explicit because we cannot get a closed form of the relative prices $\pi^{cor}, \pi^{int}$ as in Proposition \ref{main1}.\footnote{When $\alpha>0$, we have the uniqueness of $\pi^{cor}, \pi^{int}$. When $\alpha>0$, the uniqueness may fail.}

We next consider  a CARA function. Assume that $U^A(c_1,c_2)=\frac{e^{-\alpha_1c_1}}{-\alpha_1}+\gamma\frac{e^{-\alpha_2c_2}}{-\alpha_2}$ where $\alpha_1>0,\alpha_2<0,\gamma>0$. This is a strictly concave function. As proved in Appendix \ref{cara}, the demand function for good 1 of agent $A$ is 
 \begin{align}
\label{cara1} c^A_1=\begin{cases}0 &\text{ if }\alpha_2\big(e^A_1+\frac{p_2}{p_1}e^A_2\big)\leq \frac{p_2}{p_1}\Big(ln(\gamma)-ln\big(\frac{p_2}{p_1}\big)\Big)\\
 e^A_1+\frac{p_2}{p_1}e^A_2 &\text{ if }\alpha_1\big(e^A_1+\frac{p_2}{p_1}e^a_2\big)+ln(\gamma)-ln\big(\frac{p_2}{p_1}\big)\leq 0\\
 c_1^* &\text{ if } c_1^*\in (0,e^A_1+\frac{p_2}{p_1}e^A_2)\\
 \end{cases}
 \end{align}
 where $c_1^*\equiv \frac{\alpha_2\big(e^A_1+\frac{p_2}{p_1}e^A_2\big)-\frac{p_2}{p_1}\Big(ln(\gamma)-ln\big(\frac{p_2}{p_1}\big)\Big)}{\alpha_2+\frac{p_2}{p_1}\alpha_1}$. Note that three cases in (\ref{cara1}) are mutually exclusive. As in the CRRA case, we can apply Proposition \ref{mainlemma} to have a similar result but the relative prices are implicit.

{\bf Utility function of agent $B$}. The utility function $U^B(c_1,c_2)=\frac{c_1^2}{2} + \mathcal{D}\ln(c_2)$  generates a demand correspondence which is simple and explicit, and  violates standard assumptions in the existence theorems. A natural question is what happens when we change the function $U^B$? In Appendix \ref{cara}, we compute the demand correspondences for the utility function $U^B(c_1,c_2)=\frac{e^{\alpha_1c_1}}{\alpha_1}+\mathcal{D}\frac{e^{\alpha_2c_2}}{\alpha_2}$ where $\alpha_1\not=0,\alpha_2\not=0,\mathcal{D}>0$. But it is not really tractable. 

A more tractable utility function is  $u^B(c_1,c_2)=\frac{c_1^{\alpha_1}}{\alpha_1}+\mathcal{D}\frac{c_2^{\alpha_2}}{\alpha_2}$ where $\alpha_2>1>\alpha_1, \mathcal{D}>0$. This function is neither concave nor convex on $\rr_+^2$, and it is more general than $ U^B(c_1,c_2)=\frac{c_1^2}{2} + \mathcal{D}\ln(c_2)$. 

Define $f(c_1)\equiv \frac{c_1^{\alpha_1}}{\alpha_1}+\frac{\mathcal{D}}{\alpha_2}\big(\frac{w^B-p_1c_1}{p_2}\big)^{\alpha_2}$. We have
\begin{align*}
f'(c_1)=(w^B-p_1c_1)^{\alpha_2-1}\Big(c_1^{\alpha_1-1}(w^B-p_1c_1)^{1-\alpha_2}-\frac{Dp_1}{p_2^{\alpha_2}}\Big).
\end{align*}
By using the same method in Lemma \ref{demand1}, we can extend our analysis for the case $\alpha_1+\alpha_2=2$ because in this case, we can explicitly compute the demand correspondences whose forms are similar to those in Lemma \ref{agentB}. Then, we can apply Proposition \ref{main1}. When $\alpha_1+\alpha_2\not=2$, the demand correspondences are not explicit.\footnote{For instance, if $M\equiv \max_{c_1\in [0,w^B/p_1]}c_1^{\alpha_1-1}(w^B-p_1c_1)^{1-\alpha_2}\leq \frac{Dp_1}{p_2^{\alpha_2}}$, then the optimal $c^B_1=0$. If $M>\frac{Dp_1}{p_2^{\alpha_2}}$, then there are two values $x_1,x_2$ such that $0<x_2<x_1<w/p_1$,  $f'(x_1)=f'(x_2)=0$, $f'(c_1)<0 \forall c_1\in (0,x_2)\cup (x_1,w^B/p_1)$ and $f'(c_1)>0\forall c_1\in (x_1,x_2)$. Then $\max_{c_1\in [0,w^B/p_1]} f(c_1)=\max\{f(0),f(x_1)\}$ and we can use the same argument as in Lemma \ref{demand1} to write the demand correspondences.}

\subsection{Conclusion}\label{conclusion}
We have considered an exchange economy in which the utility function of one agent is neither quasiconcave nor quasiconvex. In this setting, we have explicitly characterized a necessary and sufficient condition—based on fundamental elements such as agents' endowments and preferences—for the (non-)existence of a general equilibrium. To the best of our knowledge, our results do not follow from existing existence theorems in the literature.

Furthermore, we analyze the role of the agents' endowments and preferences in determining equilibrium existence. Given that our analysis relies on specific utility functions, we view this work as an initial step toward addressing the broader and challenging issue of equilibrium (non)existence in more general models with non-convex preferences.

\appendix
\section{Appendix}
\label{appendix}
\subsection{Proof of Lemma \ref{agentB} }
Lemma \ref{agentB} is a direct consequence of the following result.
\begin{lemma}\label{demand1}
The demands for good $1$ and $2$ of agent $B$ are given by
\begin{align}
\label{cb1}c^B_1=\begin{cases}\{0\} &\text{ if } w^2_B\leq 4\mathcal{D}p_1^2\\
\{0\}& \text{ if }w^2_B> 4\mathcal{D}p_1^2 \text{ and } V(w_B, p_1) < \mathcal{D} \ln (w_B)\\
 \left \{0, \frac{w_B+\sqrt{w^2_B-4\mathcal{D}p_1^2}}{2p_1}\right \}& \text{ if }w^2_B> 4\mathcal{D}p_1^2 \text{ and } V(w_B, p_1) = \mathcal{D} \ln (w_B)\\
 \left \{\frac{w_B+\sqrt{w^2_B-4\mathcal{D}p_1^2}}{2p_1}\right \} &\text{ if } w^2_B> 4\mathcal{D}p_1^2 \text{ and } V(w_B, p_1) > \mathcal{D} \ln (w_B)
\end{cases}
\end{align}
and $c^B_2=(w_B-p_1c^B_1)/p_2$, 
where  $V(w_B, p_1)\equiv\frac{1}{2}  \big(\frac{w_B+\sqrt {w^2_B -4 \mathcal{D} p_1 ^2}}{2 p _1} \big)^2 + \mathcal{D} \ln \big(\frac{w_B-\sqrt {w^2_B -4 \mathcal{D} p_1 ^2}}{2} \big).$
\end{lemma}

\begin{proof}[{\bf Proof of Lemma \ref{demand1}}]
The budget constraint of agent $B$ is $p_1c^B_1+p_2c^B_2=w_B$. To simplify notations, we ignore subscript $B$. At optimum, the budget constraint is binding $p_1c_1+p_2c_2=w$, which implies that $c_2=\frac{w-p_1c_1}{p_2}$. So, we consider the following problem
\begin{align}
\ma_{0 \leq c_1 \leq w/p_1}\frac{c_1^2}{2} + \mathcal{D}\ln(w-p_1c_1)
\end{align}
in order to determine the demand for good $1$ of agent A. 

Let us denote $f(c_1)\equiv \frac{c_1^2}{2} + \mathcal{D}\ln(w-p_1c_1)$. We have that $f'(c_1)=\frac {-p_1 c_1 ^2 + c_1 w -\mathcal D p_1}{w -p_1 c_1}$. Observe that $f'(c_1)=0\Leftrightarrow p_1c_1^2-w c_1+{\mathcal D}p_1=0$.

Denote $
x_1 \equiv \frac{w+\sqrt{w^2-4\mathcal{D}p_1^2}}{2p_1}  \hbox{ and  } x_2\equiv \frac{w-\sqrt{w^2-4\mathcal{D}p_1^2}}{2p_1}$.

We consider all possible cases.

(1) If $w^2- 4\mathcal{D} p _1 ^2 <0$, then $f^\prime (c_1) < 0, \forall c_1 \geq 0$. So, the demand is $c_1=0$.

(2) If $w^2- 4\mathcal{D} p _1 ^2 =0$, then $c_1=x_1=x_2$ is the unique solution to the equation $f'(c_1)=0$. Moreover, we have $f^\prime (c_1)= \frac{-p_1 (c_1-\frac{w}{2p_1})^2}{w-p_1 c_1} \leq 0$. By consequence, the demand is again $c_1=0$.

(3) If $w^2- 4\mathcal{D} p _1 ^2 >0$, then we have $0<x_2<x_1<w/p_1$. Moreover, we observe that: $f^\prime (c_1) > 0 \Leftrightarrow x_2<c_1< x_1$. In this case, we have  $$\ma_{0 \leq c_1 \leq w/p_1}f(c_1)=\max\Big\{f(0),f(x_1)\Big\}.$$

Observe that $f(x_1)-f(0)=V(w,p_1)-\mathcal{D}ln(w)$. From these properties, we obtain the demand  as in Lemma \ref{demand1}.

\end{proof}

We now prove Lemma \ref{agentB}.  We look at the conditions: 
$(w_B)^2>4\mathcal{D}p_1^2$ and  ${V}(w_B,p_1) \geq \mathcal{D}\ln (w_B)$. As in step 1, we have
 \begin{align}
(w_B)^2> 4\mathcal{D}p_1^2& \Leftrightarrow e^B_1+e^B_2\frac{p_2}{p_1}> 2\sqrt{\mathcal{D}} \Leftrightarrow e^B_1+p e^B_2> 2\sqrt{\mathcal{D}}\\
{V}(w_B,p_1) \geq \mathcal{D}\ln (w_B)& \Leftrightarrow g\Big(\frac{e^B_1+e^B_2\frac{p_2}{p_1}}{\sqrt{\mathcal{D}}}\Big)\geq g(x^*) \Leftrightarrow e^B_1+pe^B_2\geq x^*\sqrt{\mathcal{D}}
    \end{align}
    because the function $g$ is strictly increasing on $ [2,\infty)$.
    By combining with Lemma \ref{demand1} and the fact that \blue{$x^*=2.2>2$}, we obtain Lemma \ref{agentB}. 



\subsection{Proof of Proposition \ref{main1}}
\begin{proof}[{\bf Proof of Proposition \ref{main1}}] 

{\bf Step 1}. Assume that there is a corner equilibrium where $c^B_1=0$. We have $c^A_1=\frac{p_1e^A_1+p_2e^A_2}{2p_1}, \quad c^A_2=\frac{p_1e^A_1+p_2e^A_2}{2p_2}$ and $c^B_2=\frac{p_1e^B_1+p_2e^B_2}{p_2}$. According to market clearing condition $c^A_1+c^B_1=e^A_1+e^B_1$, we have 
\begin{align}
\frac{p_1e^A_1+p_2e^A_2}{2p_1}=e^B_1+e^A_1   \Longleftrightarrow 
\frac{p_2}{p_1}=\frac{2e^B_1+e^A_1}{e^A_2}=\pi^{cor}
\end{align}
According to Lemma \ref{agentB}, we have $e^B_1+pe^b_2\leq x^*\sqrt{\mathcal{D}}$, i.e., $e^B_1+ \pi^{cor} e^b_2\leq x^*\sqrt{\mathcal{D}}$. 


{\bf Step 2}: Assume that there is an interior equilibrium with $c^B_1>0$. Lemma \ref{demand1} implies that $(p_1e^B_1+p_2e^B_2)^2>4\mathcal{D}p_1^2$ and  ${V}(w_B,p_1) \geq  \mathcal{D}\ln (w_B)$.

According to Lemma \ref{demand1}, the market clearing condition $c^A_1+c^B_1=e^A_1+e^B_1$ becomes $
\sqrt{(p_1e^B_1+p_2e^B_2)^2-4\mathcal{D}p_1^2}=p_1(e^B_1+e^A_1)-p_2(e^B_2+e^A_2)$. 
Since $(p_1e^B_1+p_2e^B_2)^2>4\mathcal{D}p_1^2$, this equation is equivalent to $F(X)=0$ and $
(e^B_1+e^A_1)-X(e^B_2+e^A_2)>0$, where we denote
$X\equiv p_2/p_1$. 

It means that the equation $F(X)=0$ has at least one solution in $(0,\frac{e^B_1+e^A_1}{e^B_2+e^A_2})$.  By combing this with (\ref{rank}) and the facts that $F(0)=(e^B_1+e^A_1)^2-(e_1^B)^2+4\mathcal{D}>0$ and $F(\frac{e^B_1+e^A_1}{e^B_2+e^A_2})=4\mathcal{D}-(e^B_2\frac{e^B_1+e^A_1}{e^B_2+e^A_2}+e^B_1)^2$, we get that $\Delta \geq 0$ and  $e^B_1+\frac{e^B_1+e^A_1}{e^B_2+e^A_2}e^B_2-2\sqrt{\mathcal{D}}>0$. 
So, we find that  $X=\pi^{int}$.

We now look at the conditions: 
$(w_B)^2>4\mathcal{D}p_1^2$ and  ${V}(w_B,p_1) \geq \mathcal{D}\ln (w_B)$. As in step 1, we have
 \begin{align}
(w_B)^2> 4\mathcal{D}p_1^2& \Leftrightarrow e^B_1+e^B_2\frac{p_2}{p_1}> 2\sqrt{\mathcal{D}} \Leftrightarrow e^B_1+e^B_2\pi^{int}> 2\sqrt{\mathcal{D}} \\
{V}(w_B,p_1) \geq \mathcal{D}\ln (w_B)& \Leftrightarrow g(e^B_1+e^B_2\frac{p_2}{p_1})\geq g(x^*(\mathcal D)) \Leftrightarrow e^B_1+e^B_2\pi^{int}\geq x^*(\mathcal D).
    \end{align}
Finally, we get that $e^B_1+e^B_2\pi^{int}\geq x^*{(\mathcal D)}$.


{\bf Step 3}. We will prove that: if (\ref{pcor}) is satisfied, there exists a unique equilibrium, the relative price is $\frac{p_2}{p_1}=\pi^{cor} \equiv  \frac{2e^B_1+e^A_1}{e^A_2}$ and  $c^{B}_1=0$.
First, suppose that there is another equilibrium relative price. According to step 2,  it must be $\pi^{int}$ and $\Delta \geq 0$. In this case, we have $e^B_1+e^B_2\pi^{int}\geq x^*\sqrt{\mathcal{D}}$. Since $\pi^{int}<\pi^{cor}$, we get that $e^B_1+e^B_2\pi^{cor}\geq x^*\sqrt{\mathcal{D}}$, a contradiction. Therefore, we obtain the uniqueness of equilibrium relative price.

Second, we prove that $\pi^{cor}$ is an equilibrium relative price. Indeed, let $\frac{p_2}{p_1}=\pi^{cor} \equiv  \frac{2e^B_1+e^A_1}{e^A_2}$. Since $e^B_1+e^B_2\frac{2e^B_1+e^A_1}{e^A_2}\leq x^*\sqrt{\mathcal{D}}$, we have that $(c^B_1,c^B_2)=(0,{p_1e^B_1}/{p_2}+e^B_2)$ is an optimal solution {to } agent B's maximization problem. It is easy to see that $(\frac{p_1e^A_1+p_2e^A_2}{2p_1}, \frac{p_1e^A_1+p_2e^A_2}{2p_2})$ is the unique solution {to } agent A's maximization problem.

The market clearing condition   for good $1$ is
\begin{align}
\frac{p_1e^A_1+p_2e^A_2}{2p_1} + c^B_1=e^B_1+e^A_1   
\end{align}
which implies that $c^B_1=0$ is the unique good 1 consumption of agent $B$ {at} equilibrium.

{\bf Step 4}. We can prove point 2 of Proposition \ref{main1} by adopting a similar argument used in step 3. Point 3 of of Proposition \ref{main1} is a direct consequence of points 1 and 2.

\end{proof}


\subsection{Other proofs}
\begin{proof}[{\bf Proof of Remark \ref{remark}}] Consider an equilibrium $(p_1,p_2,c^{B}_1,c^{B}_2,c^{D}_1,c^{D}_2)$. Since the utility function of agent $D$ is $U^D(c_1^D,c^D_2)=c^D_2$, we have $c_1^D=0, p_2c^D_2=w_D\equiv p_1e^{D}_1+p_2e^D_2$. The good $1$  market clearing condition implies that $c^B_1=e_1\equiv e^B_1+e^D_1>0$. Since the good 1 consumption of agent $B$ is positive, Lemma \ref{demand1} implies that $(p_1e^B_1+p_2e^B_2)^2>4\mathcal{D}p_1^2$ and  ${V}(w_B,p_1) \geq  \mathcal{D}\ln (w_B)$. Denote $X\equiv p_2/p_1$. Observe that
 \begin{align}
(w_B)^2> 4\mathcal{D}p_1^2& \Leftrightarrow e^B_1+e^B_2\frac{p_2}{p_1}> 2\sqrt{\mathcal{D}} \Leftrightarrow e^B_1+e^B_2X> 2\sqrt{\mathcal{D}} \\
{V}(w_B,p_1) \geq \mathcal{D}\ln (w_B)& \Leftrightarrow g(e^B_1+e^B_2\frac{p_2}{p_1})\geq g(x^*(\mathcal D)) \Leftrightarrow e^B_1+e^B_2X\geq x^*(\mathcal D).
    \end{align}
We now determine the relative price $X\equiv p_2/p_1$. The good 1 market clearing condition implies that $\sqrt{(e^B_1+e^B_2X)^2-4\mathcal{D}}=(e^B_1+2e^D_1)-e^B_2X$.  
Since $e^B_1+e^B_2X> 2\sqrt{\mathcal{D}}$, this is equivalent to
\begin{subequations}
\begin{align*}
&\begin{cases}
 (e^B_1+e^B_2X)^2-4\mathcal{D}=((e^B_1+2e^D_1)-e^B_2X)^2\\
(e^B_1+2e^D_1)-e^B_2X>0
\end{cases}
\Leftrightarrow
 e^B_2X=\frac{e^D_1(e^B_1+e^D_1)+\mathcal D}{e^B_1+e^D_1}<e^B_1+2e^D_1
\end{align*}
\end{subequations}
To sum up, parameters must satisfy
\begin{align*}
\begin{cases}
\frac{e^D_1(e^B_1+e^D_1)+\mathcal D}{e^B_1+e^D_1}<e^B_1+2e^D_1\\
e^B_1+e^B_2X\geq x^*(\mathcal D)
\end{cases}
\Leftrightarrow
\begin{cases}
 \mathcal D<e^B_1+e^D_1\\
e^B_1+e^D_1+\frac{\mathcal D}{e^B_1+e^D_1}\geq x^*(\mathcal D)
\end{cases}
\end{align*}
Conversely, we can easily check that under these conditions, there exists an equilibrium whose relative price $X$ is determined by 
$e^B_2X=\frac{e^D_1(e^B_1+e^D_1)+\mathcal D}{e^B_1+e^D_1}$.

\end{proof}

\subsection{Demand correspondences under CARA utility}\label{cara}

Assume that $U(c_1,c_2)=\frac{e^{\alpha_1c_1}}{\alpha_1}+\mathcal{D}\frac{e^{\alpha_2c_2}}{\alpha_2}$ where $\alpha_1\not=0,\alpha_2\not=0,\mathcal{D}>0$. We want to maximize this function subject to constraints
\begin{align}
c_1\geq 0,c_2\geq 0, \quad 
p_1c_1+p_2c_2=w
\end{align}
where $p_1>0,p_2>0,w>0$.

In optimal, we have $c_2=\frac{w-p_1c_1}{p_2}$. So, the problem becomes to maximize the function $f(c_1)\equiv \frac{e^{\alpha_1c_1}}{\alpha_1}+\mathcal{D}\frac{e^{\alpha_2\frac{w-p_1c_1}{p_2}}}{\alpha_2}$ subject to $0\leq c_1\leq \frac{w}{p_1}$.
We have
\begin{align*}
f'(c_1)=&e^{\alpha_1c_1}-\frac{\mathcal{D}p_1}{p_2}e^{\alpha_2\frac{w-p_1c_1}{p_2}}=e^{\alpha_2\frac{w-p_1c_1}{p_2}}\Big(e^{(\alpha_1+\alpha_2\frac{p_1}{p_2})c_1-\alpha_2\frac{w}{p_2}}-\frac{\mathcal{D}p_1}{p_2}\Big).
\end{align*}
Observe that $f'(c_1)\geq 0$ if and only if $(\alpha_1+\alpha_2\frac{p_1}{p_2})c_1\geq \alpha_2\frac{w}{p_2}+ln\big(\frac{\mathcal{D}p_1}{p_2}\big)$.

When $(\alpha_1+\alpha_2\frac{p_1}{p_2})\not=0$, we denote \begin{align}
c_1^*\equiv \frac{ln\big(\frac{\mathcal{D}p_1}{p_2}\big)}{(\alpha_1+\alpha_2\frac{p_1}{p_2})}.
\end{align}
There are three different cases.
\begin{enumerate}
\item $\alpha_1+\alpha_2\frac{p_1}{p_2}=0$. In this case, there are three cases.
\begin{enumerate}
\item $\alpha_2\frac{w}{p_2}+ln\big(\frac{\mathcal{D}p_1}{p_2}\big)=0$. We have $c_1= [0,w/p_1]$.

\item $\alpha_2\frac{w}{p_2}+ln\big(\frac{\mathcal{D}p_1}{p_2}\big)<0$. We have $f'(c_1)> 0$ for any $c_1$ and, hence, $c_1= \{w/p_1\}$.

\item $\alpha_2\frac{w}{p_2}+ln\big(\frac{\mathcal{D}p_1}{p_2}\big)>0$. We have $f'(c_1)< 0$ for any $c_1$, and, hence, $c_1= \{0\}$.
\end{enumerate}
\item $\alpha_1+\alpha_2\frac{p_1}{p_2}>0$. In this case, there are three cases.
\begin{enumerate}
\item $c_1^*\leq 0$. We have $f'(x)> 0 \forall c\in (0,w/p_1)$. So, we get $c_1=\{w/p_1\}$.
\item $c_1^*\geq w/p_1$. We have $f'(x)< 0 \forall c\in (0,w/p_1)$. So, we get $c_1=\{0\}$.
\item $c_1^*\in (0, w/p_1)$. We have $f'(x)< 0 \forall x\in (0,c_1^*)$ and $f'(x)> 0 \forall x\in (c_1^*,w/p_1)$. Hence, we have \begin{align}
c_1=\begin{cases}\{0\} &\text{ if } f(0)>f(w/p_1) \\
  \{0,  w/p_1\}& \text{ if }f(0)=f(w/p_1)\\
 \{w/p_1\} &\text{ if }f(0)<f(w/p_1).
\end{cases}
\end{align}

\end{enumerate}

\item $\alpha_1+\alpha_2\frac{p_1}{p_2}<0$. In this case, there are three cases. 
\begin{enumerate}
\item If $c_1^*\leq 0$, then $f'(x)< 0 \forall c\in (0,w/p_1)$. So, we get $c_1=\{0\}$.
\item If $c_1^*\geq w/p_1$, then $f'(x)> 0 \forall c\in (0,w/p_1)$. So, we get $c_1=\{w/p_1\}$.
\item If $c_1^*\in (0, w/p_1)$, then $f'(x)> 0 \forall x\in (0,c_1^*)$ and $f'(x)< 0 \forall x\in (c_1^*)$. So, we get $
c_1=\{c_1^*\}$

\end{enumerate}
\end{enumerate}

\end{document}